\documentclass[journal,onecolumn,11pt]{IEEEtran}
\usepackage{psfrag,url,amsmath,graphicx,amssymb,jeffe,stfloats,subfigure}

\newtheorem{theorem}{Theorem}
\newtheorem{lemma}{Lemma}

\newcommand{\Perp}{\perp \! \! \! \perp}
\newenvironment{definition}[1][Definition]{\begin{trivlist}
\item[\hskip \labelsep {\bfseries #1}]}{\end{trivlist}}
\IEEEoverridecommandlockouts
\begin{document}
\title{The Degraded Poisson \\Wiretap Channel\thanks{Amine Laourine and Aaron B. Wagner are with the School of Electrical and Computer Engineering, Cornell University, Ithaca,
NY 14853 USA. (Email: al496@cornell.edu, wagner@ece.cornell.edu.). Part of this paper will be presented in the 2010 IEEE International Symposium on Information Theory (ISIT 2010).}}

\author{Amine Laourine and Aaron B. Wagner}

\maketitle

\begin{abstract}
Providing security guarantees for wireless communication is critically important for today's applications. While previous work in this area has concentrated on radio frequency (RF) channels, providing security guarantees for RF channels is inherently difficult because they are prone to rapid variations due small scale fading. Wireless optical communication, on the other hand, is inherently more secure than RF communication due to the intrinsic aspects of the signal propagation in the optical and near-optical frequency range. In this paper, secure communication over wireless optical links is examined by studying the secrecy capacity of a direct detection system. For the degraded Poisson wiretap channel, a closed-form expression of the secrecy capacity is given. A complete characterization of the general rate-equivocation region is also presented. For achievability, an optimal code is explicitly constructed by using the structured code designed by Wyner for the Poisson channel. The converse is proved in two different ways: the first method relies only on simple properties of the conditional expectation and basic information theoretical inequalities, whereas the second method hinges on the recent link established between minimum mean square estimation and mutual information in Poisson channels.
\end{abstract}
\newpage
\section{Introduction}
\PARstart{P}{rotecting} information flow from unauthorized access is of vital importance for today's applications. The concept of information secrecy however is not new and dates back to the pioneering work of Shannon \cite{Shannon}. Shannon considered a strong notion of secrecy requiring essentially that the eavesdropper's received signal be independent of the communicated message. Under this constraint Shannon showed that the transmitter must share a private key with the legitimate receiver whose entropy is at least as large as the message's entropy. This negative result lead to the development of modern cryptographic systems whose objective is not to provide strong secrecy but rather to make security breaches computationally prohibitive. These encryption algorithms are employed today in most systems where secrecy is required. This wide deployment, however, should not obscure the fact that these methods do not offer strong security guarantees. Information-theoretical security can provide such guarantees and as such has attracted considerable attention lately.

The pessimistic result of Shannon is due to the implicit assumption that the eavesdropper observes the same signal as the legitimate receiver. In his celebrated paper \cite{Wyner1}, Wyner challenged this assumption by introducing the wiretap channel, a channel in which an eavesdropper observes a degraded version of the signal received by the legitimate receiver. Wyner also considered a different notion of secrecy which requires that asymptotically the equivocation of the eavesdropper about the transmitted message should converge to the message's entropy. With this new framework, and for discrete memoryless channels (DMCs), Wyner gave a complete description of the tradeoff between the information rate at the legitimate receiver and the information leaked to the eavesdropper. Wyner showed that for a DMC there exists an intrinsic quantity called the \emph{secrecy capacity} which gives the maximum information rate that can be reliably transmitted to the legitimate receiver with zero leakage of information to the eavesdropper. Since then, the wiretap channel has become one of the core topics in information theoretic security and several results are available today. Csiszar and Korner \cite{Csiszar} extended Wyner's study to nondegraded DMCs. The degraded Gaussian wiretap channel was studied in \cite{Leung}. Multiple input multiple output (MIMO) Gaussian wiretap channels have been the subject of an extensive interest lately \cite{Khisti1}-\cite{Liu}. Wiretap channels in the presence of fading have been also investigated in \cite{Gopala}-\cite{Khisti2} (and the references therein).

The results of \cite{Csiszar} show that a non-zero secrecy capacity requires that the legitimate receiver's channel be less noisy than the eavesdropper's.
For RF channels, however, this is difficult to guarantee in practice due
to the possibility of multipath fading. Indeed, even if the legitimate
receiver is closer to the transmitter than the eavesdropper is, the legitimate
receiver may still have a weaker channel due to fading. Moreover, since the
fading state is a sensitive function of the position of the receivers, it
is difficult to predict the degree of fading experienced by the legitimate
receiver and the eavesdropper given imperfect information about their
locations.

One possible solution to this problem is to use optical or near-optical frequencies instead of RF. For optical wireless systems, the detector is usually multiple orders of magnitude larger than the wavelength of the transmitted beam, which provides natural immunity against multipath fading via spatial diversity~\cite{Barry}.
This immunity makes predictions about the quality of the legitimate receiver's and the eavesdropper's signal based on their position more accurate. In fact, with the multipath problem gone, the only major channel impairment remaining is the pathloss which can be safely assumed to be higher for the eavesdropper if the legitimate receiver can guarantee that he is closer to the transmitter.

Another advantage of optical communications over RF is that the transmitted
signal is highly directional, making interception by a malevolent third party more difficult. This should be contrasted with the relatively-omnidirectional nature of RF transmissions, for which the signal is broadcasted over a wide angle. Yet another advantage of wireless optical communication is the spatial confinement of the transmitted optical signal. Indeed, at optical wavelengths, the transmitted signal is absorbed by the atmosphere and beyond a certain range it becomes undetectable. This is desirable from a security standpoint as an eavesdropper located beyond this range is literally kept in the dark.

All of these features give wireless optical communications a clear advantage for security. This technology is already being deployed in the form of infrared communications \cite{Kahn}-\cite{Green}, and ultraviolet (UV) systems are currently under development \cite{Xu}-\cite{Adee}. However, despite the numerous benefits that this technology offers, coding is still needed for those scenarios in which the channel itself does not provide absolute secrecy, such as when the beam of light is reflected or scattered by solid objects, dust or water droplets \cite{Sidorovich}. Although in this case the signal has been degraded, an eavesdropper could still gain valuable information.

We examine the fundamental limits of coding for secure communication over
optical channels by studying the secrecy capacity of the Poisson channel,
a common model for direct detection optical communications systems. In such systems the transmitter sends information by modulating the intensity of an optical signal while the receiver observes the arrival moments of individual photons. The capacity of this channel has been determined under peak power constraint on the transmitted optical power by Kabanov \cite{Kabanov} and under both average and peak power constraint by Davis \cite{Davis}. Wyner \cite{Wyner2} derived the reliability function of this channel for all rates below capacity and constructed exponentially optimal codes. Multiple-access Poisson channels were studied in \cite{Lapidoth2}-\cite{Bross} whereas broadcast Poisson channels were considered in \cite{Lapidoth1} and \cite{Sokolovsky}. The capacity of the Poisson channel has been also investigated in the presence of fading \cite{Chakraborty}.

We study in this paper the degraded Poisson wiretap channel. The legitimate receiver observes a doubly stochastic Poisson process with instantaneous rate $A_{y}X_{t}+\lambda_{y}$ where $\{X_{t},0\leq t\leq T\}$ is the signal transmitted. The eavesdropper's observation is also a doubly stochastic Poisson process with instantaneous rate $A_{z}X_{t}+\lambda_{z}$. For degradedness we assume that\footnote{These conditions were shown to be sufficient for degradedness in \cite{Lapidoth1}. The argument is reproduced in Lemma 1 below.} $A_{y}\geq A_{z}$ and $\lambda_{y}\leq \frac{A_{y}}{A_{z}}\lambda_{z}$. In Theorem 1 we provide a closed form expression of the secrecy capacity as a function of the parameters $(A_{u},\lambda_{u}), \text{ } u\in\{y,z\}$. This result is further extended by Theorem 5 which gives a full characterization of the rate equivocation region.

Our achievability proof uses stochastic encoding as well as the structured codes constructed by Wyner for the Poisson channel \cite{Wyner2}. As for the converse, we will see that the infinite bandwidth nature of the Poisson channel makes it possible to prove the converse using only simple properties of the conditional expectation combined with basic information theoretical inequalities. This is to be contrasted with the converse of the (finite bandwidth) Gaussian channel which is proved using the entropy power inequality (EPI). As an illustration for the basic ideas that underpin the converse for the Poisson channel, we will start by considering here the more familiar infinite bandwidth Gaussian channel and we will see also that for this channel the proof of the converse simplifies considerably.

For this purpose, consider the continuous time Gaussian wiretap channel with bandwidth $B$ (later we will let $B$ tend to infinity) and with a power constraint $P$. This continuous time channel is equivalent to $2B$ uses per second of the discrete time Gaussian channel depicted in Fig. 1. The input signal is power constrained, i.e., $\mathbb{E}[X^{2}]\leq P$, the legitimate receiver observes $Y=X+W_{1}$ and the eavesdropper receives $Z=Y+W_{2}=X+W_{1}+W_{2}$, where $W_{i}\sim\mathcal{N}(0,N_{i}B)$ and $W_{1}\Perp W_{2}$.

\begin{figure}[htb]
\begin{psfrags}
  \psfrag{a}{$W_{2}$}
  \psfrag{b}{$Z$}
  \psfrag{c}{$X$}
  \psfrag{d}{$W_{1}$}
  \psfrag{e}{$Y$}
           \includegraphics{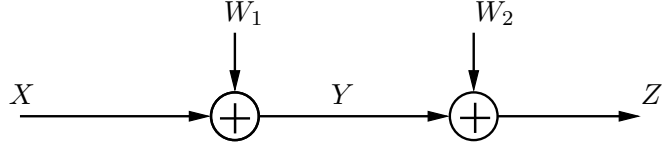}
      \caption{The discrete time Gaussian wiretap channel}
      \end{psfrags}
\end{figure}

Define $\tilde{N}$ by $\frac{1}{\tilde{N}}=\frac{1}{N_{1}}-\frac{1}{N_{1}+N_{2}}$ and observe that
\begin{align}\label{gaussian}
\frac{\tilde{N}}{N_{1}}Y=X+\frac{N_{1}}{N_{2}}Z+\tilde{W},
\end{align}
where $\tilde{W}=W_{1}-\frac{N_{1}}{N_{2}}W_{2}$. It is easy to see that $\tilde{W}\sim\mathcal{N}(0,\tilde{N}B)$ and $\mathbb{E}[\tilde{W}(W_{1}+W_{2})]=0$, it follows therefore that $\tilde{W}\Perp W_{1}+W_{2}$ (since they are jointly Gaussian). For the discrete time Gaussian wiretap channel, it is known that the secrecy capacity is given by $\max_{p_{X}}(I(X;Y)-I(X;Z))$. For the continuous time channel counterpart with bandwidth $B$, the secrecy capacity becomes $C_{s}^{B}=2B\max_{p_{X}}(I(X;Y)-I(X;Z))$.

In \cite{Leung}, using the celebrated EPI, a closed form expression for the secrecy capacity of the discrete time Gaussian wiretap channel was obtained. In just a few steps, we will see that the secrecy capacity of the infinite bandwidth ($B\rightarrow\infty$) Gaussian wiretap channel can be found much more simply. Starting with (\ref{gaussian}) we obtain the following sequence of inequalities
\begin{align}
I(X;Y)&=I(X;X+\frac{N_{1}}{N_{2}}Z+\tilde{W})\nonumber\\
&\stackrel{(a)}{\leq} I(X;X+\tilde{W},Z) \nonumber\\
&\stackrel{(b)}{=}h(X+\tilde{W},Z)-h(X+\tilde{W},Z|X)\nonumber\\
&\stackrel{(c)}{\leq} h(X+\tilde{W})+h(Z)-h(X+\tilde{W},Z|X)\nonumber\\
&\stackrel{(d)}=h(X+\tilde{W})+h(Z)-h(X+\tilde{W}|X)-h(Z|X)\nonumber\\
&\stackrel{(e)}{=}I(X;X+\tilde{W})+I(X;Z),\nonumber
\end{align}
Inequality $(a)$ follows from the data processing inequality, equalities in $(b)$ and $(e)$ are standard information theory identities, $(c)$ follows from the independence bound on entropy and finally $(d)$ holds because $X+\tilde{W}$ and $X+W_{1}+W_{2}$ are conditionally independent given $X$ (i.e., $(X+\tilde{W})\Perp (X+ W_{1}+W_{2})|X$). Basically, the key identity needed to go from $(a)$ to $(e)$ is the following: if $Y_{1}\Perp Y_{2}|X$, then we have $I(X;Y_{1},Y_{2})\leq I(X;Y_{1})+I(X;Y_{2})$. A proof of this simple inequality in a more general setting will be given later and will be used in part of the converse for the Poisson channel. Going back to the Gaussian problem, we see that
\begin{align}
C_{s}^{B}=2B\max_{P_{X}}(I(X;Y)-I(X;Z))\leq 2B\max_{P_{X}}I(X;X+\tilde{W})=B\ln(1+\frac{P}{B\tilde{N}}).
\end{align}
For a fixed bandwidth $B$, this last inequality is not tight. Now letting $B\rightarrow\infty$ we obtain
\begin{align}
C_{s}^{\infty}\triangleq\lim_{B\rightarrow\infty} C_{s}^{B}\leq \frac{P}{\tilde{N}}=\frac{P}{N_{1}}-\frac{P}{N_{1}+N_{2}}.
\end{align}
However, since $\max_{p_{X}}(I(X;Y)-I(X;Z))\geq\max_{p_{X}}I(X;Y)-\max_{p_{X}}I(X;Z)$, we also have that
\begin{align}
C_{s}^{\infty}\geq \lim_{B\rightarrow\infty}\left(B\ln(1+\frac{P}{BN_{1}})-B\ln(1+\frac{P}{B(N_{1}+N_{2})})\right)=\frac{P}{N_{1}}-\frac{P}{N_{1}+N_{2}}.
\end{align}
It follows that $C_{s}^{\infty}=\frac{P}{N_{1}}-\frac{P}{N_{1}+N_{2}}$.

This remarkably simple approach will be particularly useful for the Poisson channel. More specifically, when $\frac{\lambda_{y}}{A_{y}}=\frac{\lambda_{z}}{A_{z}}$, the eavesdropper's signal $Z$ is a thinned version of the legitimate receiver's signal $Y$, i.e., \footnote{The time dependence has been dropped to ease the notations. Refer to the converse part of the paper for a mathematically precise statement.} $Y=Z+\tilde{Z}$ where $Z\Perp\tilde{Z}|X$, the approach above gives that $I(X;Y)-I(X;Z)\leq I(X;\tilde{Z})$. Since $\tilde{Z}$ is itself a doubly stochastic Poisson process, the mutual information $I(X;\tilde{Z})$ can be maximized using the martingale techniques of Kabanov \cite{Kabanov} and an (achievable) upperbound can be obtained on $I(X;Y)-I(X;Z)$. When $\frac{\lambda_{y}}{A_{y}}<\frac{\lambda_{z}}{A_{z}}$, a different bounding technique using only simple properties of the conditional expectation will be devised.

Although no ``sophisticated" tools are required to prove the converse, we show in the appendix that using some new results in information theory an alternative proof can be provided. This different proof hinges on the link that has been established between the mutual information (MI) and the minimum mean square estimation (MMSE) in Poisson channels \cite{Guo}. It is worth noting at this point that the link between the MI and the MMSE in the Gaussian setting \cite{Guo2} has been also used recently for different Gaussian wiretap channels \cite{Ekrem}, \cite{Bustin}.

One of the distinctive aspects in this paper is that we do not resort to the $\Delta$-discretization method introduced by Wyner \cite{Wyner2}. This method was used to approximate the Poisson channel by a binary DMC thereby allowing the transposition of the widely known results for DMCs to the Poisson channel. This technique leads to extensive computations, especially when we are interested in the secrecy capacity as there are now two conflicting objectives involved, the maximization of the information rate at the legitimate receiver and the minimization of the information leakage at the eavesdropper. We circumvent the use of this method by using the techniques described above.

The rest of this paper is organized as follows. The next section describes the setup of the problem and
presents the main result of this paper as well as some interpretations of the obtained result.
The proof of the achievability of the secrecy capacity is given in Section III and the proof of the converse is presented in Section IV. In Section V we extend the main result of the paper by giving a complete characterization of the rate-equivocation region. Finally, in section VI, some possible future directions are discussed.
\section{Problem and Result Statement}
The input process to the Poisson channel is a waveform denoted by $X_{0}^{T}\triangleq \{X_{t},0\leq t\leq T\}$ satisfying $X_{t}\geq0$ for all $t$. We further assume that the input process
is peak power limited, i.e., $X_{t}\leq1$ for all $t$. The received signal at the legitimate receiver $Y_{0}^{T}$ is a doubly stochastic Poisson process with instantaneous rate $A_{y}X_{t}+\lambda_{y}$, i.e., given $X_{0}^{T}$ the stochastic process $Y_{0}^{T}$ has independent increments with $Y_{0}=0$ and for $0\leq s\leq t\leq T$ we have
\[\Pr(Y_{t}-Y_{s}=k|X_{0}^{T})=\frac{1}{k!}\Upsilon^{k}(s,t)e^{-\Upsilon(s,t)},\text{ }k\in \mathbb{N},\]
where
\[\Upsilon(s,t)=\int_{s}^{t}(A_{y}X_{\tau}+\lambda_{y})d\tau.\]
The parameter $A_{y}>0$ accounts for possible signal attenuation at the receiver. The parameter
$\lambda_{y}\geq0$ is the dark current intensity which results from background noise and bears no information on the input process $X_{0}^{T}$.
Similarly the output process of the eavesdropper $Z_{0}^{T}$ is a doubly stochastic Poisson process with
instantaneous rate $A_{z}X_{t}+\lambda_{z}$.


In this paper, the space of doubly stochastic Poisson processes on the interval $[0,T]$ will be denoted by
$\mathcal{P}(T)$. Following the notation used in \cite{Guo} the output process of the Poisson channel in the
interval $[0,T]$ with instantaneous rate $\alpha X_{t}+\lambda$ will be denoted by $\mathcal{P}_{0}^{T}(\alpha
X_{0}^{T}+\lambda)$. We use $\langle X_{t}\rangle_{s}$ to designate $E[X_{t}|\mathcal{P}_{0}^{s}(X_{0}^{s})]$,
as such $\langle X_{t}\rangle_{t}$ refers to the causal conditional mean estimate and $\langle X_{t}\rangle_{T}$
to the noncausal one.

All stochastic processes considered in this paper are defined on a common measurable space
$(\Omega,\mathcal{F})$. We use $\mathcal{F}_{\xi}^{s}$ to denote the internal history generated by the process
$\xi_{0}^{s}$.

In this paper we are interested in the degraded Poisson wiretap channel. Lapidoth et al.
\cite{Lapidoth1} gave conditions on the parameters $(A_{u},\lambda_{u}), \text{ } u\in\{y,z\}$ for stochastic
degradedness. These conditions are presented in the following lemma. In order to prepare for the results to come
we will also briefly go over the proof of this lemma.
\begin{lemma}[Lapidoth, Telatar and Urbanke \cite{Lapidoth1}]
The eavesdropper's channel is stochastically degraded with respect to the legitimate receiver's channel, if
\begin{equation}\label{intensity}
A_{y}\geq A_{z},
\end{equation}
and
\begin{equation}\label{dark}
\lambda_{y}\leq \frac{A_{y}}{A_{z}}\lambda_{z}.
\end{equation}
\end{lemma}
\begin{proof}
Let $\tilde{Y}_{0}^{T}$ (cf. Fig. 2) be the process defined as follows
\begin{equation}\label{Ytilda}
\tilde{Y}_{t}=Y_{t}+H_{t},\text{ } t\in[0,T],
\end{equation}
where $H_{t}$ is a homogeneous Poisson process with rate $\tilde{\lambda}=\frac{A_{y}}{A_{z}}\lambda_{z}-\lambda_{y}$ (note that $\tilde{\lambda}\geq0$ by (\ref{dark})) independent of $(X_{0}^{T},Y_{0}^{T})$. It follows that $\tilde{Y}_{0}^{T}$ is a doubly stochastic Poisson process with instantaneous rate $A_{y}X_{t}+\lambda_{y}+\tilde{\lambda}=A_{y}X_{t}+\frac{A_{y}}{A_{z}}\lambda_{z}$. The process $Z_{0}^{T}$ is then obtained from $\tilde{Y}_{0}^{T}$ by thinning with erasure probability $1-\frac{A_{z}}{A_{y}}$ (note that because of (\ref{intensity}) this quantity is $\geq0$).
\end{proof}
In the rest of this paper we will assume that at least one of the inequalities (\ref{intensity}) or (\ref{dark}) is strict. Note that this assumption can be made without losing generality for if there was an equality in (\ref{intensity}) and (\ref{dark}) then the legitimate receiver's channel and the eavesdropper's channel will be identical and the secrecy capacity will be zero.

\begin{figure}[htb]
\begin{center}
\begin{psfrags}
  \psfrag{a}{$X_{0}^{T}$}
  \psfrag{b}{\small{$\mathcal{P}_{0}^{T}(A_{y}X_{0}^{T}+\lambda_{y})$}}
  \psfrag{c}{$Y_{0}^{T}$}
  \psfrag{d}{$\mathcal{P}_{0}^{T}(\tilde{\lambda})$}
  \psfrag{e}{$\tilde{Y}_{0}^{T}$}
  \psfrag{f}{\small{\text{Thinning $(\frac{A_{z}}{A_{y}})$}}}
  \psfrag{g}{$Z_{0}^{T}$}
         \includegraphics{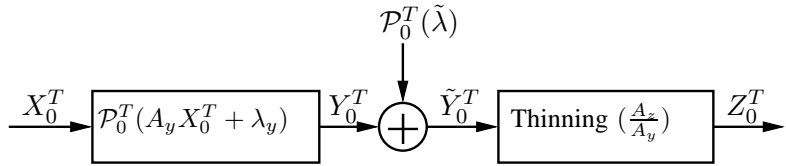}
      \caption{The degraded Poisson wiretap channel}
      \end{psfrags}
\end{center}
\end{figure}

We move now to the description of the information transmission aspect of the problem. The transmitter wishes to
communicate a message $U$ uniformly distributed on $\mathcal{U}=\{1,...,M\}$. An $(M,T)$ code $(E_{T},D_{T})$
for the Poisson wiretap channel is a stochastic encoder $E_{T}$ that maps a message $U$ to a waveform
$X_{0}^{T}$ which satisfies the peak power constraint and a decoder $D_{T}:\mathcal{P}(T)\rightarrow\mathcal{U}$.
The transmission rate of this code is
\[R=\frac{H(U)}{T}=\frac{1}{T}\ln M.\]
The average probability of error at the legitimate receiver is
\begin{equation}
P_{e}=\frac{1}{M}\sum_{m=1}^{M}\Pr\left(D_{T}(Y_{0}^{T})\neq m|U=m\right).
\end{equation}
The level of secrecy in this paper is measured by $\frac{1}{T}I(U;Z_{0}^{T})$. This normalized mutual information quantifies the amount of information about the message $U$
leaked to the eavesdropper. As such our goal is to make this quantity as small as possible.
\begin{definition}
A secrecy rate $R_{s}$ is said to be \textit{achievable} for the Poisson wiretap channel if for all $\epsilon>0$ and all
sufficiently large $T$, there exists an $(M,T)$ code such that
\begin{align}
\frac{\ln M}{T}&\geq R_{s}-\epsilon\nonumber\\
P_{e}&\leq\epsilon\nonumber\\
\frac{1}{T}I(U;Z_{0}^{T})&\leq\epsilon
\end{align}
\end{definition}
The supremum of achievable secrecy rates will be called the \textit{secrecy capacity}. The main result of this paper is
the following.
\begin{theorem} The secrecy capacity of the degraded Poisson wiretap channel is given by\footnote{If $\lambda=0$, the convention is that $0^0=1$.}
\begin{equation}\label{secrecy}
C_{s}=\alpha^{*}(A_{y}-A_{z})+\ln\left(\frac{\lambda_{y}^{\lambda_{y}}}{\lambda_{z}^{\lambda_{z}}}\right)+\ln\left(\frac{(A_{z}\alpha^{*}+\lambda_{z})^{\lambda_{z}}}{(A_{y}\alpha^{*}+\lambda_{y})^{\lambda_{y}}}\right),
\end{equation}
where $\alpha^{*}$ is the unique solution in $[0,1]$ to the following equation
\begin{equation}\label{equation}
\frac{(A_{y}\alpha^{*}+\lambda_{y})^{A_{y}}}{(A_{z}\alpha^{*}+\lambda_{z})^{A_{z}}}=e^{A_{z}-A_{y}}\frac{(A_{y}+\lambda_{y})^{A_{y}+\lambda_{y}}}{(A_{z}+\lambda_{z})^{A_{z}+\lambda_{z}}}\frac{\lambda_{z}^{\lambda_{z}}}{\lambda_{y}^{\lambda_{y}}}.
\end{equation}
\end{theorem}

This result assumes that $(A_z,\lambda_z)$ is known to the transmitter.
Yet it follows that $C_s$ is an achievable rate with perfect
secrecy even if the eavesdropper observes
$\mathcal{P}_{0}^{T}(A^\prime_z X_{0}^{T}+\lambda_z^\prime)$,
where $A_z^\prime$ and $\lambda_z^\prime$ are unknown but satisfy
$A_z^\prime \le A_z$ and $\lambda_z^\prime \ge \frac{A_z^\prime}{A_z}\lambda_z$. Thus, only
one-sided estimates of $A_z$ and $\lambda_z$ are needed. In practice,
an upper bound on $A_z$ could be provided by guaranteeing that any
potential eavesdropper is more than a certain distance away from
the transmitter. A lower bound on the dark current $\lambda_z$ could
be provided using ambient noise measurements and the known physical
limitations of existing receivers.

\textbf{Worst case scenario:} A particularly insightful case is when $\frac{\lambda_{y}}{A_{y}}=\frac{\lambda_{z}}{A_{z}}=\sigma$. This situation happens when the eavesdropper observes a thinned version of the signal of the legitimate receiver, i.e., $H_{t}\equiv0$ in (\ref{Ytilda}). In this case, after some algebraic manipulations, we obtain that
\begin{equation}
\alpha^{*}=\frac{(1+\sigma)^{1+\sigma}}{e\sigma^{\sigma}}-\sigma,
\end{equation}
and the secrecy capacity reduces to
\begin{equation}
C_{s}=(\lambda_{y}-\lambda_{z})\left(\frac{1}{e}\left(1+\frac{1}{\sigma}\right)^{1+\sigma}-(1+\sigma)\ln\left(1+\frac{1}{\sigma}\right)\right).
\end{equation}
This is saying that the secrecy capacity is the difference between the capacity of the main channel (the channel between the transmitter and the legitimate receiver) and the capacity of the eavesdropper's channel. For instance, in the special case when there is no dark current $\lambda_{y}=\lambda_{z}=0$, we find that $\alpha^{*}=\frac{1}{e}$ and the secrecy capacity reduces to
\begin{equation}
C_{s}=\frac{A_{y}-A_{z}}{e}.
\end{equation}
For a degraded DMC, Wyner \cite{Wyner1} showed that the secrecy capacity is equal to $\max_{p_{X}}(I(X;Y)-I(X;Z))$. Hence the following inequality is always satisfied
\[\text{Secrecy Capacity}\geq C_{M}-C_{W},\]
where $C_{M}$ is the capacity of the main channel and $C_{W}$ is the capacity of the eavesdropper's channel. As shown in \cite{Van}, there is equality in the inequality above if there is an input probability distribution $p_{X}$ that maximizes simultaneously $I(X;Y)$ and $I(X;Z)$. This is exactly what is happening here, when $\frac{\lambda_{y}}{A_{y}}=\frac{\lambda_{z}}{A_{z}}$ the mutual information $I(X_{0}^{T};Y_{0}^{T})$ and $I(X_{0}^{T};Z_{0}^{T})$ are both maximized by letting the input $X_{0}^{T}$ cycle infinitely fast between its extreme values, i.e., the peak power $1$ and $0$ with $\Pr(X_{t}=1)=1-\Pr(X_{t}=0)=\alpha^{*}=\frac{(1+\sigma)^{1+\sigma}}{e\sigma^{\sigma}}-\sigma$.

Before we proceed further with the presentation of the problem considered in this paper, we give a lemma that will prove to be useful in the proofs of the achievability and the converse, a proof of this result can be found for instance in \cite{Kabanov}.
\begin{lemma} The mutual information between the input $X_{0}^{T}$ and the output $\mathcal{P}_{0}^{T}(\alpha
X_{0}^{T}+\lambda)$ can be upper bounded as follows\footnote{Note that some authors use the function $\vartheta(x)=(\alpha x+\lambda)\ln(\alpha x+\lambda)-\lambda\ln\lambda$ instead but the constant term $\lambda\ln\lambda$ cancels out here.}
\begin{equation}
I(X_{0}^{T};\mathcal{P}_{0}^{T}(\alpha
X_{0}^{T}+\lambda))\leq\int_{0}^{T}(\mathbb{E}[\vartheta(X_{t})]-\vartheta(\mathbb{E}[X_{t}]))dt,
\end{equation}
where $\vartheta(x)=(\alpha x+\lambda)\ln(\alpha x+\lambda)$.
\end{lemma}
\section{Achievability of $C_{s}$}
Our achievability proof relies on the structured codes that were designed for the Poisson channel by Wyner \cite{Wyner2}.
Before delving into the details of the proof, we will briefly describe the code construction and the properties inherited by this code.

\textbf{Wyner codes $\mathcal{W}(T,M,k)$:} Let $T$, $M$ and $k$ be given, and construct an $M\times{{M}\choose{k}}$ binary matrix $\mathcal{X}$ as follows. The columns of $\mathcal{X}$ are the ${{M}\choose{k}}$ binary $M$-vectors with exactly $k$ ones and $M-k$ zeros.
Now partition the interval $[0,T]$ into ${{M}\choose{k}}$ subintervals of equal length $\varpi_{T}\triangleq\frac{T}{{{M}\choose{k}}}$ and construct $M$ waveforms $\{X_{0}^{T}(m)\}_{m=1}^{M}$ as follows
\begin{equation}
X_{t}(m)=\mathcal{X}(m,n),\text{} t\in((n-1)\varpi_{T},n\varpi_{T}], \text{ } n=1,...,{{M}\choose{k}}.
\end{equation}
For $\alpha=\frac{k}{M}$ fixed, these codes satisfy
\begin{equation}\label{code}
\frac{1}{T}\mu\{t:X_{t}(m)=1\}=\alpha,\quad \text{for all } m,
\end{equation}
with $\mu$ being the Lebesgue measure. If moreover $M=e^{RT}$, for $T>>1$, Wyner showed that for $m\neq m'$
\begin{equation}\label{code2}
\frac{1}{T}\mu\{t:X_{t}(m)=1,X_{t}(m')=0\}\approx\alpha(1-\alpha).
\end{equation}
As such for $T$ large enough the codewords $\{X_{0}^{T}(m)\}_{m=1}^{M}$ will behave as if they were chosen independently.

After this brief overview of Wyner codes we are in a position to state the achievability theorem and prove it.
\begin{theorem}
Any secrecy rate $R_{s}<C_{s}$ is achievable.
\end{theorem}
\begin{proof}
Let $\epsilon>0$ be arbitrary and let $R_{s}=C_{s}-\epsilon$. Define
\begin{align}
R_{u}&=\alpha^{*}(A_{u}+\lambda_{u})\ln(A_{u}+\lambda_{u})+(1-\alpha^{*})\lambda_{u}\ln\lambda_{u}\nonumber\\&-(A_{u}\alpha^{*}+\lambda_{u})\ln(A_{u}\alpha^{*}+\lambda_{u}), \text{ }u\in\{y,z\}.
\end{align}
After few algebraic manipulations, we can show that
\[C_{s}=R_{y}-R_{z}.\]
Given these parameters, the encoder-decoder pair considered here works as follows.

\textbf{Encoding:} Let $M=e^{R_{s}T}$ and let $U$ be uniformly distributed on $\mathcal{U}=\{1,...,M\}$. Define $M_{y}=e^{(R_{y}-\frac{3}{2}\epsilon)T}$ and following the steps described above construct a code\footnote{Note that even if $\alpha^{*}$ is not a rational number, it can be approximated arbitrary close by rationals.} $\mathcal{C}=\mathcal{W}(T,M_{y},\alpha^{*}M_{y})$. Partition this code arbitrarily into $M$ smaller subcodes, i.e., $\mathcal{C}=\cup_{i=1}^{M}\mathcal{C}_{i}$. The cardinality of each each subcode $\mathcal{C}_{i}$ will be equal to $M_{z}=\frac{M_{y}}{M}=e^{(R_{z}-\frac{\epsilon}{2})T}$.

The encoder works as follows, when the message $U=m$ is chosen, the codeword $X_{0}^{T}$ is selected uniformly randomly from $\mathcal{C}_{m}$.

\textbf{Decoding:} The decoder considered here is the maximum likelihood decoder constructed by Wyner \cite{Wyner2}. After observing $Y_{0}^{T}$, the decoder at the legitimate receiver computes the following metric
\begin{equation}
\Psi_{m}=\int_{S_{m}}dY_{t},
\end{equation}
where $S_{m}=\{t\in[0,T]:\text{ }X_{t}(m)=1\}$. Then $D_{T}(Y_{0}^{T})=m$ if $m$ maximizes $\Psi_{m}$, with ties resolved in favor of the smallest $m$.

\textbf{Analysis of $P_{e}$:} The fact that $P_{e}\rightarrow0$, follows simply from the fact that Wyner codes with the peak power $1$ and average power $\alpha^{*}$ are capacity achieving.

\textbf{Analysis of $\frac{1}{T}I(U;Z_{0}^{T})$:}

Notice first that for each $m$, the waveform $X_{0}^{T}(m)$ is piecewise constant. It follows that a sufficient statistic for making a decision is the number of arrivals during each subinterval $((n-1)\varpi_{T},n\varpi_{T}]$, i.e., $Z_{n}=Z_{n\varpi_{T}}-Z_{(n-1)\varpi_{T}}$, $n=1,...,N_{y}$ with $N_{y}={{M_{y}}\choose{\alpha^{*}M_{y}}}$.
Consequently,
\begin{align}
I(X_{0}^{T};Z_{0}^{T})&=I(\mathbf{X};\mathbf{Z})\nonumber\\
I(U;Z_{0}^{T})&=I(U;\mathbf{Z})
\end{align}
where $\mathbf{X}=[X_{1},...,X_{N_{y}}]$, $X_{i}=0$ or $1$ depending on the choice of the codeword and $\mathbf{Z}=[Z_{1},...,Z_{N_{y}}]$. The equalities above follows from the fact that $\mathbf{Z}$ is a sufficient statistic.

As a result of Lemma 2, we have
\begin{equation}
\frac{1}{T}I(X_{0}^{T};Z_{0}^{T})\leq\frac{1}{T}\int_{0}^{T}(\mathbb{E}[\phi_{z}(X_{t})]-\phi_{z}(\mathbb{E}[X_{t}]))dt,
\end{equation}
where $\phi_{z}(x)=(A_{z}x+\lambda_{z})\ln(A_{z}x+\lambda_{z})$. Because of the uniform choice in the encoding scheme and in view of (\ref{code}) we must have that $\Pr[X_{t}=1]=1-\Pr[X_{t}=0]=\alpha^{*}$, hence we have
\begin{align}
\mathbb{E}[\phi_{z}(X_{t})]&=\alpha^{*}\phi_{z}(1)+(1-\alpha^{*})\phi_{z}(0)\nonumber\\
&=\alpha^{*}(A_{z}+\lambda_{z})\ln(A_{z}+\lambda_{z})+(1-\alpha^{*})\lambda_{z}\ln\lambda_{z}.
\end{align}
and
\begin{equation}
\phi_{z}(\mathbb{E}[X_{t}])=\phi_{z}(\alpha^{*})=(A_{z}\alpha^{*}+\lambda_{z})\ln(A_{z}\alpha^{*}+\lambda_{z}).
\end{equation}
Consequently, we deduce that
\begin{equation}\label{ineq1}
\frac{1}{T}I(\mathbf{X};\mathbf{Z})=\frac{1}{T}I(X_{0}^{T};Z_{0}^{T})\leq R_{z}.
\end{equation}
Notice that every subcode $\mathcal{C}_{m}$ can be viewed as a code for the eavesdropper's channel with $M_{z}$ codewords and uniform prior distribution. Define $\delta_{m}$ to be the probability of error for code $\mathcal{C}_{m}$ ($1\leq m\leq M$) with the (optimal) decoder described above and let $\delta=\frac{1}{M}\sum_{m=1}^{M}\delta_{m}$. From Fano's inequality we have
\begin{equation}
H(\mathbf{X}|\mathbf{Z},U=m)\leq H(\delta_{m})+\delta_{m}\ln M_{z},
\end{equation}
where $H(p)=-p\ln(p)-(1-p)\ln(1-p)$ is the binary entropy.

Since the codewords are uniformly distributed in each subcode, we deduce that $H(\mathbf{X}|U=m)=\ln M_{z}$. We conclude therefore that
\begin{align}
I(\mathbf{X};\mathbf{Z}|U=m)&=H(\mathbf{X}|U=m)-H(\mathbf{X}|\mathbf{Z},U=m)\nonumber\\
&\geq \ln M_{z}-(H(\delta_{m})+\delta_{m}\ln M_{z}).
\end{align}
Averaging over $U$ and by using the concavity of $H(\cdot)$ we find that
\begin{equation}
I(\mathbf{X};\mathbf{Z}|U)\geq\ln M_{z}-(H(\delta)+\delta\ln M_{z}).
\end{equation}
Notice also that $U\rightarrow \mathbf{X}\rightarrow \mathbf{Z}$ form a Markov chain, i.e.,
\begin{align}
\frac{1}{T}I(U;\mathbf{Z})&=\frac{1}{T}I(U,\mathbf{X};\mathbf{Z})-\frac{1}{T}I(\mathbf{X};\mathbf{Z}|U)\nonumber\\
&=\frac{1}{T}I(\mathbf{X};\mathbf{Z})-\frac{1}{T}I(\mathbf{X};\mathbf{Z}|U)
\end{align}
Combined with the last inequality this implies that
\begin{equation}\label{ineq2}
\frac{1}{T}I(U;\mathbf{Z})\leq\frac{1}{T}I(\mathbf{X};\mathbf{Z})-\frac{1}{T}\ln M_{z}+\frac{1}{T}(H(\delta)+\delta\ln M_{z}).
\end{equation}
Inequalities (\ref{ineq1}) and (\ref{ineq2}) result in the following
\begin{equation}\label{ineq6}
\frac{1}{T}I(U;\mathbf{Z})\leq R_{z}-\frac{1}{T}[\ln M_{z}-(H(\delta)+\delta\ln M_{z})].
\end{equation}
As $M_{z}=e^{(R_{z}-\frac{\epsilon}{2})T}$, this gives
\begin{equation}
\frac{1}{T}I(U;Z_{0}^{T})=\frac{1}{T}I(U;\mathbf{Z})\leq\frac{\epsilon}{2}+\frac{1}{T}H(\delta)+\delta(R_{z}-\frac{\epsilon}{2}).
\end{equation}
By the code construction described above, the codewords of every subcode $\mathcal{C}_{m}$ ($1\leq m\leq M$) satisfy (\ref{code}) and (\ref{code2}) (with $\alpha$ replaced by $\alpha^{*}$). These two conditions dictate the pairwise error probability of the codewords in $\mathcal{C}_{m}$ \cite{Wyner2}. Since the overall error probability of the code $\mathcal{C}_{m}$ is governed by the pairwise error probability \cite{Wyner2}, it follows that every subcode $\mathcal{C}_{m}$ is capacity achieving for the eavesdropper's channel and as such $\delta_{m}$ for $m=1,...,M$ can be made arbitrarily small. Hence, by choosing $T$ large enough, we can enforce that $\frac{1}{T}H(\delta)+\delta(R_{z}-\frac{\epsilon}{2})\leq\frac{\epsilon}{2}$. The previous inequality shows therefore that $\frac{1}{T}I(U;Z_{0}^{T})\leq\epsilon$ and the desired secrecy condition is satisfied.

This shows that any secrecy rate $R_{s}<C_{s}$ can be achieved and completes the achievability proof.
\end{proof}
\section{The Converse for the secrecy capacity}
Before delving into the details of the converse we need the following technical lemma due to Wyner \cite{Wyner3}.
\begin{lemma}[Wyner \cite{Wyner3}] If $\Gamma:\Omega\rightarrow\digamma$ is a random variable such that $\digamma$ is a finite set and $\Lambda_{0}^{T}=\{\Lambda_{t},0\leq t\leq T\}$ is a given stochastic process, then we have
\begin{equation}
I(\Gamma;\Lambda_{0}^{T})=H(\Gamma)-H(\Gamma|\Lambda_{0}^{T}),
\end{equation}
where $H(\Gamma)$ is the usual entropy for discrete random variables and
\begin{equation}
H(\Gamma|\Lambda_{0}^{T})=-\mathbb{E}\left[\sum_{\gamma\in\digamma}\Pr[\Gamma=\gamma|\mathcal{F}^{T}_{\Lambda}]\ln\Pr[\Gamma=\gamma|\mathcal{F}^{T}_{\Lambda}]\right].
\end{equation}
\end{lemma}
This lemma is standard when all the random variables have discrete alphabets however this extension is needed in
this paper since we are dealing with continuous time stochastic processes.

The converse theorem will be proved through a sequence of Lemmas. The first one gives an inequality
that must satisfied by every encoder-decoder pair $(E_{T},D_{T})$.
\begin{lemma}
For every $(M,T)$ code with rate $R=\frac{\ln M}{T}$ we have
\begin{equation}
R\leq \frac{1}{T(1-P_{e})}\left(I(X_{0}^{T};Y_{0}^{T}|Z_{0}^{T})+I(U;Z_{0}^{T})+H(P_{e})\right).
\end{equation}
\end{lemma}
\begin{proof}
Let $\hat{U}=D_{T}(Y_{0}^{T})$ denote the output of the decoder at the legitimate receiver, so that $P_{e}=\Pr(U\neq\hat{U})$.
We then have the following sequence of identities
\begin{align}
RT=\ln M=H(U)&\stackrel{(a)}{=}H(U|Y_{0}^{T})+I(U;Y_{0}^{T})\nonumber\\
&\stackrel{(b)}{\leq} H(U|\hat{U})+I(U;Y_{0}^{T})\nonumber\\
&\stackrel{(c)}{\leq} H(P_{e})+P_{e}\ln M+I(U;Y_{0}^{T}),
\end{align}
the equality $(a)$ follows from Wyner's lemma and the inequality $(c)$ is an application of Fano's inequality. For the inequality $(b)$, since $U\rightarrow Y_{0}^{T}\rightarrow\hat{U}$ is a Markov chain we deduce that\footnote{The data processing inequality extends to arbitrary random variables, see for instance Theorem 3.4 in \cite{Wyner3}.} $I(U,Y_{0}^{T})\geq I(U,\hat{U})$. Now by invoking Wyner's lemma again it follows that $H(U|Y_{0}^{T})\leq H(U|\hat{U})$.

From Kolmogorov's formula (see Lemma 3.2 in \cite{Wyner3}) we have\footnote{The definition of the conditional mutual information for arbitrary random variables can be found in \cite{Wyner3}.}
\begin{equation}
I(U;Y_{0}^{T},Z_{0}^{T})=I(U;Y_{0}^{T})+I(U;Z_{0}^{T}|Y_{0}^{T}),
\end{equation}
since $U\rightarrow Y_{0}^{T}\rightarrow Z_{0}^{T}$ is a Markov chain we deduce that\footnote{Refer to Lemma 3.1. in \cite{Wyner3}.} $I(U;Z_{0}^{T}|Y_{0}^{T})=0$. By applying Kolmogorov's formula again we obtain
\begin{align}
I(U;Y_{0}^{T})=I(U;Y_{0}^{T},Z_{0}^{T})&=I(U;Z_{0}^{T})+I(U;Y_{0}^{T}|Z_{0}^{T})\nonumber\\
&\leq I(U;Z_{0}^{T})+I( X_{0}^{T};Y_{0}^{T}|Z_{0}^{T}),
\end{align}
where the last inequality follows from the fact that $U\rightarrow X_{0}^{T}\rightarrow Y_{0}^{T}\rightarrow
Z_{0}^{T}$ form a Markov chain. Combining this last inequality with $(c)$ and rearranging the terms yields the
desired inequality.
\end{proof}
\begin{lemma} If $\mathbb{E}\int_{0}^{T}|X_{t}\ln X_{t}|<\infty$, then
\begin{equation}
I( X_{0}^{T};Y_{0}^{T}|Z_{0}^{T})=I(X_{0}^{T};Y_{0}^{T})-I( X_{0}^{T};Z_{0}^{T})
\end{equation}
\end{lemma}
\begin{proof}
Applying Kolmogorov's formula twice gives
\begin{align}
I( X_{0}^{T};Y_{0}^{T},Z_{0}^{T})&=I( X_{0}^{T};Z_{0}^{T})+I( X_{0}^{T};Y_{0}^{T}|Z_{0}^{T})\nonumber\\
&=I( X_{0}^{T};Y_{0}^{T})+I( X_{0}^{T};Z_{0}^{T}|Y_{0}^{T}).
\end{align}
Since $X_{0}^{T}\rightarrow Y_{0}^{T}\rightarrow Z_{0}^{T}$ form a Markov chain, we have $I( X_{0}^{T};Z_{0}^{T}|Y_{0}^{T})=0$. Consequently we deduce that
\begin{equation}
I( X_{0}^{T};Y_{0}^{T})=I( X_{0}^{T};Z_{0}^{T})+I( X_{0}^{T};Y_{0}^{T}|Z_{0}^{T}).
\end{equation}
The condition $\mathbb{E}\int_{0}^{T}|X_{t}\ln X_{t}|<\infty$ implies that $I( X_{0}^{T};Z_{0}^{T})<\infty$ and it follows that $I( X_{0}^{T};Y_{0}^{T})-I( X_{0}^{T};Z_{0}^{T})=I( X_{0}^{T};Y_{0}^{T}|Z_{0}^{T})$.
\end{proof}

The goal of the upcoming lemmas is to prove that $I(X_{0}^{T};Y_{0}^{T}|Z_{0}^{T})\leq T C_{s}$, where $C_{s}$ is given by (\ref{secrecy}). We first decompose $I(X_{0}^{T};Y_{0}^{T}|Z_{0}^{T})$ as follows
\begin{align}\label{decomposition}
I(X_{0}^{T};Y_{0}^{T}|Z_{0}^{T})&=I(X_{0}^{T};Y_{0}^{T})-I( X_{0}^{T};Z_{0}^{T})\nonumber\\
&=I(X_{0}^{T};Y_{0}^{T})-I(X_{0}^{T};\tilde{Y}_{0}^{T})\nonumber\\&+I(X_{0}^{T};\tilde{Y}_{0}^{T})-I(X_{0}^{T};Z_{0}^{T}),
\end{align}
where $\tilde{Y}_{0}^{T}$ has been defined in (\ref{Ytilda}). The next two lemmas will provide upper bounds on $I(X_{0}^{T};Y_{0}^{T})-I(X_{0}^{T};\tilde{Y}_{0}^{T})$ and $I(X_{0}^{T};\tilde{Y}_{0}^{T})-I(X_{0}^{T};Z_{0}^{T})$.

\begin{lemma} If $\mathbb{E}\int_{0}^{T}|X_{t}\ln X_{t}|<\infty$, then
\begin{align}
&I(X_{0}^{T};Y_{0}^{T})-I(X_{0}^{T};\tilde{Y}_{0}^{T})\leq\nonumber\\&\int_{0}^{T}\left(\frac{A_{y}}{A_{z}}(\phi_{z}(\mathbb{E}[X_{t}])-\mathbb{E}[\phi_{z}(X_{t})])-
(\phi_{y}(\mathbb{E}[X_{t}])-\mathbb{E}[\phi_{y}(X_{t})])\right)dt,
\end{align}
where $\phi_{y}(x)=(A_{y}x+\lambda_{y})\ln(A_{y}x+\lambda_{y})$ and $\phi_{z}(x)$ has been defined above analogously.
\end{lemma}
\begin{proof}
Note first that \cite{Kabanov}, \cite{Bremaud}
\begin{equation}
I(X_{0}^{T};Y_{0}^{T})=\int_{0}^{T}\left(\mathbb{E}[\phi_{y}(X_{t})]-\mathbb{E}[\phi_{y}(\mathbb{E}[X_{t}|\mathcal{F}_{Y}^{t}])]\right)dt,
\end{equation}
and
\begin{align}
I(X_{0}^{T};\tilde{Y}_{0}^{T})=\int_{0}^{T}\left(\mathbb{E}[\chi(X_{t})]-\mathbb{E}[\chi(\mathbb{E}[X_{t}|\mathcal{F}_{\tilde{Y}}^{t}])]\right)dt,
\end{align}
where $\chi(x)=(A_{y}x+\frac{A_{y}}{A_{z}}\lambda_{z})\ln(A_{y}x+\frac{A_{y}}{A_{z}}\lambda_{z})$. Consequently, using the fact that $\chi(x)=\frac{A_{y}}{A_{z}}\phi_{z}(x)+\ln(\frac{A_{y}}{A_{z}})(A_{y}x+\frac{A_{y}}{A_{z}}\lambda_{z})$ and after simplifications, we deduce the following
\begin{align}
I(X_{0}^{T};Y_{0}^{T})-I(X_{0}^{T};\tilde{Y}_{0}^{T})&=\int_{0}^{T}\left(\mathbb{E}[\phi_{y}(X_{t})]-\mathbb{E}[\phi_{y}(\mathbb{E}[X_{t}|\mathcal{F}_{Y}^{t}])]\right)dt\nonumber\\&-\frac{A_{y}}{A_{z}}\int_{0}^{T}\left(\mathbb{E}[\phi_{z}(X_{t})]-\mathbb{E}[\phi_{z}(\mathbb{E}[X_{t}|\mathcal{F}_{\tilde{Y}}^{t}])]\right)dt.
\end{align}
Recall that $\tilde{Y}_{0}^{T}=Y_{0}^{T}+H_{0}^{T}$, where $H_{0}^{T}$ is a homogeneous Poisson process independent of $(X_{0}^{T},Y_{0}^{T})$. Clearly, $\mathcal{F}_{\tilde{Y}}^{t}\subset\mathcal{F}_{Y}^{t}\vee\mathcal{F}_{H}^{t}$, with $\mathcal{F}_{Y}^{t}\vee\mathcal{F}_{H}^{t}=\sigma(\mathcal{F}_{Y}^{t}\cup\mathcal{F}_{H}^{t})$ being the smallest sigma-field containing $\mathcal{F}_{Y}^{t}\cup\mathcal{F}_{H}^{t}$. From the independence of $(X_{0}^{T},Y_{0}^{T})$ from $H_{0}^{T}$, using the law of redundant conditioning (see, e.g. \cite[pp. 281-282]{Bremaud}), we deduce that
\begin{equation}\label{egalite}
\mathbb{E}[X_{t}|\mathcal{F}_{Y}^{t}\vee\mathcal{F}_{H}^{t}]=\mathbb{E}[X_{t}|\mathcal{F}_{Y}^{t}] \quad\text{a.s.}
\end{equation}
We can now establish the following sequence of identities
\begin{align}
\mathbb{E}[\phi_{z}(\mathbb{E}[X_{t}|\mathcal{F}_{Y}^{t}])]&\stackrel{(a)}{=}\mathbb{E}[\phi_{z}(\mathbb{E}[X_{t}|\mathcal{F}_{Y}^{t}\vee\mathcal{F}_{H}^{t}])]\\
&\stackrel{(b)}{=}\mathbb{E}[\mathbb{E}[\phi_{z}(\mathbb{E}[X_{t}|\mathcal{F}_{Y}^{t}\vee\mathcal{F}_{H}^{t}])|\mathcal{F}_{\tilde{Y}}^{t}]]\\
&\stackrel{(c)}{\geq}\mathbb{E}[\phi_{z}(\mathbb{E}[\mathbb{E}[X_{t}|\mathcal{F}_{Y}^{t}\vee\mathcal{F}_{H}^{t}]|\mathcal{F}_{\tilde{Y}}^{t}])]\\
&\stackrel{(d)}{=}\mathbb{E}[\phi_{z}(\mathbb{E}[X_{t}|\mathcal{F}_{\tilde{Y}}^{t}])],
\end{align}
where $(a)$ follows from (\ref{egalite}), $(b)$ follows from the smoothing property of the conditional expectation, $(c)$ from Jensen's
inequality applied to the convex function $\phi_{z}(\cdot)$ and $(d)$ from the fact that
$\mathcal{F}_{\tilde{Y}}^{t}\subset\mathcal{F}_{Y}^{t}\vee\mathcal{F}_{H}^{t}$ and the smoothing property.

We deduce therefore that
\begin{align}
I(X_{0}^{T};Y_{0}^{T})-I(X_{0}^{T};\tilde{Y}_{0}^{T})&\leq\int_{0}^{T}\left(\mathbb{E}[\phi_{y}(X_{t})]-\mathbb{E}[\phi_{y}(\mathbb{E}[X_{t}|\mathcal{F}_{Y}^{t}])]\right)dt\nonumber\\
&-\frac{A_{y}}{A_{z}}\int_{0}^{T}\left(\mathbb{E}[\phi_{z}(X_{t})]-\mathbb{E}[\phi_{z}(\mathbb{E}[X_{t}|\mathcal{F}_{Y}^{t}])]\right)dt.
\end{align}
A simple derivation shows that the function $\pi(x)=\phi_{y}(x)-\frac{A_{y}}{A_{z}}\phi_{z}(x)$ is convex as
\begin{equation}
\pi''(x)=\frac{A_{y}(\lambda_{z}A_{y}-\lambda_{y}A_{z})}{(A_{y}x+\lambda_{y})(A_{z}x+\lambda_{z})}\geq0.
\end{equation}
Now invoking again Jensen's inequality we obtain that
\begin{align}
\mathbb{E}[\pi(\mathbb{E}[X_{t}|\mathcal{F}_{Y}^{t}])]\geq\pi(\mathbb{E}[\mathbb{E}[X_{t}|\mathcal{F}_{Y}^{t}]])=\pi(\mathbb{E}[X_{t}]).
\end{align}
Using this last inequality and after rearranging the terms we obtain the desired result, i.e.,
\begin{align}
I(X_{0}^{T};Y_{0}^{T})-I(X_{0}^{T};\tilde{Y}_{0}^{T})&\leq\int_{0}^{T}\left(\mathbb{E}[\phi_{y}(X_{t})]-\phi_{y}(\mathbb{E}[X_{t}])\right)dt\nonumber\\
&-\frac{A_{y}}{A_{z}}\int_{0}^{T}\left(\mathbb{E}[\phi_{z}(X_{t})]-\phi_{z}(\mathbb{E}[X_{t}])\right)dt.
\end{align}
\end{proof}
An alternative proof of this lemma using the link provided in \cite{Guo} between the MMSE and the mutual information in Poisson channels is given in Appendix A.
\begin{lemma} If $\mathbb{E}\int_{0}^{T}|X_{t}\ln X_{t}|<\infty$, then
\begin{align}
&I(X_{0}^{T};\tilde{Y}_{0}^{T})-I(X_{0}^{T};Z_{0}^{T})\nonumber\\&\leq(\frac{A_{y}}{A_{z}}-1)\int_{0}^{T}(\mathbb{E}[\phi_{z}(X_{t})]-\phi_{z}(\mathbb{E}[X_{t}]))dt
\end{align}
\end{lemma}
\begin{proof} Recall that $Z_{0}^{T}$ was obtained from $\tilde{Y}_{0}^{T}$ by thinning with erasure probability $1-\frac{A_{z}}{A_{y}}$. Let the process $\tilde{Z}_{0}^{T}$ denote those points that were erased, hence we have that $\tilde{Z}_{0}^{T}$ is a doubly stochastic Poisson process with instantaneous rate $(A_{y}-A_{z})X_{t}+(\frac{A_{y}}{A_{z}}-1)\lambda_{z}$. Moreover $Z_{0}^{T}$ and $\tilde{Z}_{0}^{T}$ are independent given $X_{0}^{T}$. We proceed with the proof of the lemma by showing that the following inequality holds
\begin{equation}\label{subadd}
I(X_{0}^{T};\tilde{Y}_{0}^{T})-I(X_{0}^{T};Z_{0}^{T})\leq I(X_{0}^{T};\tilde{Z}_{0}^{T}).
\end{equation}
Indeed, notice first that $X_{0}^{T}\rightarrow(Z_{0}^{T},\tilde{Z}_{0}^{T})\rightarrow\tilde{Y}_{0}^{T}$ is a Markov chain, hence from the data processing inequality we deduce that
\begin{equation}\label{data}
I(X_{0}^{T};\tilde{Y}_{0}^{T})=I(X_{0}^{T};Z_{0}^{T}+\tilde{Z}_{0}^{T})\leq I(X_{0}^{T};Z_{0}^{T},\tilde{Z}_{0}^{T}).
\end{equation}
Consider now two partitions of $\Omega$, $\mathcal{Q}_{Z}=\{A_{i}\}_{i=1}^{N_{1}}\subseteq\mathcal{F}_{Z}^{T}$ and  $\mathcal{Q}_{\tilde{Z}}=\{B_{j}\}_{j=1}^{N_{2}}\subseteq\mathcal{F}_{\tilde{Z}}^{T}$. Define two discrete random variables $D$ and $\tilde{D}$ on $\Omega$ as follows
$D(\omega)=i$ if $\omega\in A_{i}$ and $\tilde{D}(\omega)=j$ if $\omega\in B_{j}$. The mutual information $I(X_{0}^{T};Z_{0}^{T},\tilde{Z}_{0}^{T})$ can be computed as \cite{Wyner3}
\begin{equation}
I(X_{0}^{T};Z_{0}^{T},\tilde{Z}_{0}^{T})=\sup_{\mathcal{Q}_{Z},\mathcal{Q}_{\tilde{Z}}}I(X_{0}^{T};D,\tilde{D}),
\end{equation}
where the supremum is taken over all such partitions of $\Omega$.
We proceed to prove (\ref{subadd}) as follows
\begin{align}
I(X_{0}^{T};D,\tilde{D})&\stackrel{(a)}{=}H(D,\tilde{D})-H(D,\tilde{D}|X_{0}^{T})\nonumber\\
&\stackrel{(b)}{\leq} H(D)+H(\tilde{D})-H(D,\tilde{D}|X_{0}^{T})\nonumber\\
&\stackrel{(c)}{=}H(D)+H(\tilde{D})-H(D|X_{0}^{T})-H(\tilde{D}|X_{0}^{T})\nonumber\\
&\stackrel{(d)}{=}I(D;X_{0}^{T})+I(\tilde{D};X_{0}^{T}),
\end{align}
where $(a)$ follows from Lemma 3 (Wyner's lemma) applied to the random variable $(D,\tilde{D})$, $(d)$ is a also a direct instance of this lemma. The inequality $(b)$ is the independence bound on the entropy (which holds here since the random variables $D$ and $\tilde{D}$ are discrete). The equality $(c)$ results from the fact that $D$ and $\tilde{D}$ are conditionally independent given $X_{0}^{T}$, indeed $D\in\mathcal{F}_{Z}^{T}$ whereas $\tilde{D}\in\mathcal{F}_{\tilde{Z}}^{T}$ and $\mathcal{F}_{Z}^{T}$ and $\mathcal{F}_{\tilde{Z}}^{T}$ are conditionally independent given $\mathcal{F}_{X}^{T}$. Consequently we have
\begin{align}
I(X_{0}^{T};Z_{0}^{T},\tilde{Z}_{0}^{T})&=\sup_{\mathcal{Q}_{Z},\mathcal{Q}_{\tilde{Z}}}I(X_{0}^{T};D,\tilde{D})\nonumber\\
&\leq\sup_{\mathcal{Q}_{Z},\mathcal{Q}_{\tilde{Z}}}(I(X_{0}^{T};D)+I(X_{0}^{T};\tilde{D}))\nonumber\\
&=I(X_{0}^{T};Z_{0}^{T})+I(X_{0}^{T};\tilde{Z}_{0}^{T}).
\end{align}
Combining the last inequality with (\ref{data}) we deduce that\footnote{Note that since $\mathbb{E}\int_{0}^{T}|X_{t}\ln X_{t}|<\infty$, then $I(X_{0}^{T};Z_{0}^{T})<\infty$ and hence the inequality is well defined.}
\begin{equation}
I(X_{0}^{T};\tilde{Y}_{0}^{T})-I(X_{0}^{T};Z_{0}^{T})\leq I(X_{0}^{T};\tilde{Z}_{0}^{T}).
\end{equation}
Now using Lemma 2 we have
\begin{align}
I(X_{0}^{T};\tilde{Z}_{0}^{T})\leq\int_{0}^{T}(\mathbb{E}[\varphi(X_{t})]-\varphi(\mathbb{E}[X_{t}]))dt,
\end{align}
where
\begin{equation}
\varphi(x)=((A_{y}-A_{z})x+(\frac{A_{y}}{A_{z}}-1)\lambda_{z})\ln((A_{y}-A_{z})x+(\frac{A_{y}}{A_{z}}-1)\lambda_{z}).
\end{equation}
Notice now that
\begin{equation}
\varphi(x)=(\frac{A_{y}}{A_{z}}-1)\phi_{z}(x)+(\frac{A_{y}}{A_{z}}-1)\ln(\frac{A_{y}}{A_{z}}-1)(A_{z}x+\lambda_{z}).
\end{equation}
Plugging this identity in the inequality above, the linear term in $x$ disappears and we are left with the inequality presented in the lemma.
\end{proof}
An alternative proof of this lemma using the link provided in \cite{Guo} between the MMSE and the mutual information in Poisson channels is given in Appendix B.
\begin{theorem}If $\mathbb{E}\int_{0}^{T}|X_{t}\ln X_{t}|<\infty$, then
\begin{equation}
\frac{1}{T}I(X_{0}^{T};Y_{0}^{T}|Z_{0}^{T})\leq C_{s}
\end{equation}
\end{theorem}
\begin{proof} Combining (\ref{decomposition}) and the result of the two previous lemmas yields
\begin{equation}\label{keyinequality}
I(X_{0}^{T};Y_{0}^{T})-I(X_{0}^{T};Z_{0}^{T})\leq\int_{0}^{T}(\mathbb{E}[K(X_{t})]-K(\mathbb{E}[X_{t}]))dt,
\end{equation}
where $K(x)=\phi_{y}(x)-\phi_{z}(x)$. A straightforward computation shows that
\begin{equation}
K''(x)=\frac{A_{z}A_{y}(A_{y}-A_{z})x+\lambda_{z}A_{y}^{2}-\lambda_{y}A_{z}^{2}}{(A_{y}x+\lambda_{y})(A_{z}x+\lambda_{z})},
\end{equation}
since $A_{y}\geq A_{z}$ and $\lambda_{z}A_{y}^{2}\geq\lambda_{y}A_{y}A_{z}\geq\lambda_{y}A_{z}^{2}$ we deduce that $K''(x)\geq0$. Moreover due to the assumption that at least one of the inequalities (\ref{intensity}) or (\ref{dark}) is strict, we conclude that $K''(x)>0$ (for $x>0$) and $K(\cdot)$ is strictly convex.

Notice now that we have
\begin{align}\label{opt1}
&\frac{1}{T}\int_{0}^{T}(\mathbb{E}[K(X_{t})]-K(\mathbb{E}[X_{t}]))dt\nonumber\\
&\stackrel{(a)}{\leq}\max_{0\leq\alpha\leq1}\left(\max_{\rho:\int_{0}^{1}x\rho(dx)=\alpha}\int_{0}^{1}K(x)\rho(dx)-K(\alpha)\right)\nonumber\\
&\stackrel{(b)}{=}\max_{0\leq\alpha\leq1}\left(\alpha K(1)+(1-\alpha)K(0)-K(\alpha)\right),
\end{align}
where $(a)$ follows from fixing $\mathbb{E}[X_{t}]=\alpha$ and maximizing over all distributions $\rho(x)$ on $[0,1]$ with mean $\alpha$.
Equality $(b)$ follows from the convexity of $K(\cdot)$ (refer to \cite{Kabanov} and \cite{Krein}), i.e., the maximizing distribution $\rho$ puts all the mass on the extremes $\{0,1\}$ and since the mean is $\alpha$, the maximizing $\rho$ assigns the mass $\alpha$ to $1$ and $1-\alpha$ to $0$.

The maximization of the last term shows that the optimal $\alpha^{*}$ is the unique solution to the equation
\[K'(\alpha^{*})=K(1)-K(0),\]
which, after some algebraic manipulations, gives that $\alpha^{*}$ is the solution to (\ref{equation}). The existence of $\alpha^{*}$ follows simply from the mean value theorem, whereas the uniqueness is a consequence of the strict monotonicity of $K'(x)$.

Consequently, the following is true
\begin{align}\label{result}
&\frac{1}{T}\int_{0}^{T}(\mathbb{E}[K(X_{t})]-K(\mathbb{E}[X_{t}]))dt\nonumber\\
&\leq\alpha^{*} K(1)+(1-\alpha^{*})K(0)-K(\alpha^{*})\nonumber\\
&=\alpha^{*}(A_{y}-A_{z})+\ln\left(\frac{\lambda_{y}^{\lambda_{y}}}{\lambda_{z}^{\lambda_{z}}}\right)+\ln\left(\frac{(A_{z}\alpha^{*}+\lambda_{z})^{\lambda_{z}}}{(A_{y}\alpha^{*}+\lambda_{y})^{\lambda_{y}}}\right).
\end{align}
This fact when combined with (\ref{keyinequality}) gives the result announced in the theorem.
\end{proof}
We are now in a position to prove the converse theorem.
\begin{theorem}[Converse]
If $R_{s}$ is an achievable secrecy rate then $R_{s}\leq C_{s}$.
\end{theorem}
\begin{proof}
Since the secrecy rate $R_{s}$ is achievable then for all $0<\epsilon<\frac{1}{2}$ and sufficiently large $T$, there
exists an $(M,T)$ code such that $\frac{\ln M}{T}\geq R_{s}-\epsilon$, $P_{e}\leq\epsilon$ and
$\frac{1}{T}I(U;Z_{0}^{T})\leq\epsilon$. Hence we have
\begin{align}
R_{s}&\leq\frac{\ln M}{T}+\epsilon\nonumber\\
&\stackrel{(a)}{\leq}\frac{1}{T(1-P_{e})}\left(I(X_{0}^{T};Y_{0}^{T}|Z_{0}^{T})+I(U;Z_{0}^{T})+H(P_{e})\right)+\epsilon\nonumber\\
&\stackrel{(b)}{\leq}\frac{1}{1-P_{e}}\left(C_{s}+\frac{I(U;Z_{0}^{T})}{T}+\frac{H(P_{e})}{T}\right)+\epsilon\nonumber\\
&\stackrel{(c)}{\leq}\frac{1}{1-\epsilon}\left(C_{s}+\epsilon+\frac{H(\epsilon)}{T}\right)+\epsilon,
\end{align}
where inequality $(a)$ follows from Lemma 4, inequality $(b)$ from Theorem 3 and inequality $(c)$ from the properties of the code. Now since $\epsilon$ is arbitrary, letting $\epsilon\rightarrow0$ yields $R_{s}\leq C_{s}$.
\end{proof}
\section{Rate-Equivocation region}
In this section we turn our attention to the rate equivocation region of the degraded Poisson wiretap channel. The level of ignorance of the eavesdropper about the transmitted message $U$
will be measured here by the normalized equivocation given by
\begin{equation}
\Delta_{T}=\frac{H(U|Z_{0}^{T})}{H(U)}.
\end{equation}
\begin{definition}
A rate-equivocation pair $(R,d)$ is said to be \textit{achievable} for the Poisson wiretap channel if for all $\epsilon>0$ and all
sufficiently large $T$, there exists an $(M,T)$ code such that
\begin{align}
\frac{\ln M}{T}&\geq R-\epsilon\nonumber\\
P_{e}&\leq\epsilon\nonumber\\
\frac{H(U|Z_{0}^{T})}{H(U)}&\geq d-\epsilon
\end{align}
\end{definition}
The following theorem gives the rate equivocation region for the degraded Poisson Wiretap channel.
\begin{theorem}
The rate-equivocation region is the set of all rate-equivocation pairs $(R,d)$ for which there exits some $\alpha\in[0,1]$ such that
\begin{align}
Rd&\leq\alpha\ln\left(\frac{(A_{y}+\lambda_{y})^{A_{y}+\lambda_{y}}}{(A_{z}+\lambda_{z})^{A_{z}+\lambda_{z}}}\right)+(1-\alpha)\ln\left(\frac{\lambda_{y}^{\lambda_{y}}}{\lambda_{z}^{\lambda_{z}}}\right)-\ln\left(\frac{(A_{y}\alpha+\lambda_{y})^{A_{y}\alpha+\lambda_{y}}}{(A_{z}\alpha+\lambda_{z})^{A_{z}\alpha+\lambda_{z}}}\right)\\
R&\leq\alpha\ln\left((A_{y}+\lambda_{y})^{A_{y}+\lambda_{y}}\right)+(1-\alpha)\ln\left(\lambda_{y}^{\lambda_{y}}\right)-\ln\left((A_{y}\alpha+\lambda_{y})^{A_{y}\alpha+\lambda_{y}}\right)\\
d&\leq 1
\end{align}
\end{theorem}
To ease the notations, using the functions $K(\cdot)$ and $\phi_{y}(\cdot)$, we can rewrite the two first inequalities as $Rd\leq\alpha K(1)+(1-\alpha)K(0)-K(\alpha)$ and $R\leq \alpha \phi_{y}(1)+(1-\alpha)\phi_{y}(0)-\phi_{y}(\alpha)$.
\begin{proof}
The main ingredients needed to prove this theorem has been already used to obtain the secrecy capacity. More specifically, for the achievability proof we will use stochastic encoding combined with Wyner codes for the Poisson channel, and for the converse we will use the key inequality (\ref{keyinequality}) established by Lemma 6 and 7.
\subsection{Direct result}
Note first that for a fixed rate $R$, if the rate equivocation pair $(R,d)$ is achievable then the pair $(R,\tilde{d})$ is achievable for all $0\leq\tilde{d}\leq d$. Hence, in order to establish the direct result, it is enough to prove that any rate-equivocation pair $(R,d)$ satisfying $Rd=\alpha K(1)+(1-\alpha)K(0)-K(\alpha)$, $R\leq \alpha \phi_{y}(1)+(1-\alpha)\phi_{y}(0)-\phi_{y}(\alpha)$
and $d\leq 1$ for some $\alpha\in[0,1]$ is achievable.

Define
\begin{align}
R_{u}&=\alpha\phi_{u}(1)+(1-\alpha)\phi_{u}(0)-\phi_{u}(\alpha), \text{ }u\in\{y,z\}.
\end{align}
Let $\epsilon>0$ be arbitrary (small enough) and let $R=\frac{R_{y}-R_{z}-\epsilon R}{d}$ with $d\leq1$ and $R\leq \alpha \phi_{y}(1)+(1-\alpha)\phi_{y}(0)-\phi_{y}(\alpha)$. The message $U$ to be transmitted is selected uniformly randomly from $\mathcal{U}=\{1,...,M\}$ with $M=e^{RT}$. Define $M_{y}=e^{(R_{y}-3\epsilon\frac{R}{2})T}$ and, following the steps described for the achievability of the secrecy capacity, construct the Wyner code $\mathcal{C}=\mathcal{W}(T,M_{y},\alpha M_{y})$. Partition this code arbitrarily into $M$ smaller subcodes, i.e., $\mathcal{C}=\cup_{i=1}^{M}\mathcal{C}_{i}$. The cardinality of each subcode $\mathcal{C}_{i}$ will be equal to $M_{z}=\frac{M_{y}}{M}=e^{(R_{y}-R-3\epsilon\frac{R}{2})T}$. Notice that with this choice of parameters we have
\begin{align}
\frac{1}{T}\ln M_{z}=R_{y}-R-3\epsilon\frac{R}{2}\leq R_{y}-Rd-3\epsilon\frac{R}{2}=R_{z}-\epsilon\frac{R}{2}.
\end{align}
The probability of error $P_{e}$ of the legitimate receiver can be made less than $\epsilon$ because the Wyner code $\mathcal{C}$ can achieve the rate $R_{y}$.

The equivocation of the code $\mathcal{C}$ can be lower bounded using the same steps used to established the upper bound on $I(U;Z_{0}^{T})$ for the secrecy capacity, as follows
\begin{align}
\Delta_{T}&=\frac{H(U|Z_{0}^{T})}{H(U)}=1-\frac{I(U;Z_{0}^{T})}{RT}\\
&\stackrel{(a)}{\geq}1-\frac{R_{z}}{R}+\frac{1}{RT}\ln M_{z}-\frac{1}{RT}\left(H(\delta)+\delta\ln M_{z}\right)\\
&=1-\frac{R_{z}}{R}+\frac{R_{y}-R-3\epsilon\frac{R}{2}}{R}-\frac{1}{RT}\left(H(\delta)+\delta\ln M_{z}\right)\\
&\geq d-\frac{\epsilon}{2}-\frac{1}{RT}H(\delta)-\delta (\frac{R_{z}}{R}-\frac{\epsilon}{2}).
\end{align}
In the above, inequality $(a)$ follows from (\ref{ineq6}) and $\delta=\frac{1}{M}\sum_{m=1}^{M}\delta_{m}$ where $\delta_{m}$ is the probability of error for the code $\mathcal{C}_{m}$ ($1\leq m\leq M$) with the (optimal) decoder described previously.

As was discussed before, the term $\frac{1}{RT}H(\delta)+\delta (\frac{R_{z}}{R}-\frac{\epsilon}{2})$ can be made less than $\frac{\epsilon}{2}$ for $T$ large enough, which means that
\begin{equation}
\Delta_{T}=\frac{H(U|Z_{0}^{T})}{H(U)}\geq d-\epsilon.
\end{equation}
This establishes that the rate-equivocation pair $(R,d)$ is achievable.
\subsection{Converse}
For every $(M,T)$ code with rate $R_{T}=\frac{\ln M}{T}$ and equivocation $\Delta_{T}=\frac{H(U|Z_{0}^{T})}{H(U)}$ we have
\begin{align}
TR_{T}\Delta_{T}&=H(U|Z_{0}^{T})=H(U)-I(U;Z_{0}^{T})\nonumber\\
&= H(U|Y_{0}^{T})+I(U;Y_{0}^{T})-I(U;Z_{0}^{T})\nonumber\\
&\leq H(U|\hat{U})+I(U;Y_{0}^{T}|Z_{0}^{T})\nonumber\\
&\leq H(P_{e})+P_{e}\ln M+I(X_{0}^{T};Y_{0}^{T}|Z_{0}^{T}).
\end{align}
From Lemma 6 and 7 (cf. (\ref{keyinequality})) we have that
\begin{equation}
I(X_{0}^{T};Y_{0}^{T}|Z_{0}^{T})\leq\int_{0}^{T}(\mathbb{E}[K(X_{t})]-K(\mathbb{E}[X_{t}]))dt.
\end{equation}
Consequently, we deduce that
\begin{align}
R_{T}\Delta_{T}&\leq \frac{H(P_{e})+P_{e}\ln M}{T}+\frac{1}{T}\int_{0}^{T}(\mathbb{E}[K(X_{t})]-K(\mathbb{E}[X_{t}]))dt\\
&\leq \frac{H(P_{e})+P_{e}\ln M}{T}+\alpha K(1)+(1-\alpha)K(0)-K(\alpha),
\end{align}
with $\alpha=\frac{1}{T}\int_{0}^{T}\mathbb{E}[X_{t}]dt$ and the last inequality follows from the convexity of the function $K(\cdot)$. Note that since $0\leq X_{t}\leq1$ it follows that $0\leq\alpha\leq1$.

Similarly, we have that
\begin{align}
R_{T}=\frac{H(U)}{T}&=\frac{1}{T}H(U|Y_{0}^{T})+\frac{1}{T}I(U;Y_{0}^{T})\nonumber\\
&\leq \frac{1}{T}H(U|\hat{U})+\frac{1}{T}I(X_{0}^{T};Y_{0}^{T})\nonumber\\
&\stackrel{(a)}{\leq} \frac{1}{T}(H(P_{e})+P_{e}\ln M)+\frac{1}{T}\int_{0}^{T}(\mathbb{E}[\phi_{y}(X_{t})]-\phi_{y}(\mathbb{E}[X_{t}]))dt\nonumber\\
&\stackrel{(b)}{\leq} \frac{H(P_{e})+P_{e}\ln M}{T}+\alpha \phi_{y}(1)+(1-\alpha)\phi_{y}(0)-\phi_{y}(\alpha),
\end{align}
where $(a)$ follows from Fano's inequality and Lemma 2 and (b) follows from the convexity of the function $\phi_{y}(\cdot)$.

Assume now that $(R,d)$ is achievable, then for all $0<\epsilon<\frac{1}{2}$ and all sufficiently large $T$, there
exists an $(M,T)$ code such that $R_{T}\geq R-\epsilon$, $P_{e}\leq\epsilon$ and
$\Delta_{T}\geq d-\epsilon$. By definition $\Delta_{T}\leq1$, and hence $d\leq1+\epsilon$ and in light of the previous inequalities we have
\begin{align}
(R-\epsilon)(d-\epsilon)&\leq \frac{H(\epsilon)+\epsilon\ln M}{T}+\alpha K(1)+(1-\alpha)K(0)-K(\alpha)\\
(R-\epsilon)&\leq \frac{H(\epsilon)+\epsilon\ln M}{T}+\alpha \phi_{y}(1)+(1-\alpha)\phi_{y}(0)-\phi_{y}(\alpha).
\end{align}
Now since $\epsilon$ is arbitrary, letting $\epsilon\rightarrow0$ yields the desired result.
\end{proof}
\section{Conclusion and discussion}
Motivated by the practical advantages of optical communication over
RF for secure communication, we have derived
the secrecy capacity and characterized the rate-equivocation region of the degraded
Poisson wiretap channel.

Several interesting problems remain open and deserve further investigation. One is the non-degraded Poisson Wiretap
channel. One can imagine a situation in which the eavesdropper is
equipped with a powerful
detector characterized by a negligible dark current (i.e., $\lambda_{z}=0$). If the detector of the legitimate receiver has a higher received power from the
transmitter but is more prone to dark current, then
the channel will not be degraded. This is a practically-important situation but is not covered by the results of this paper.

Another issue that we have not considered is fading. Indeed, for wireless optical communications, atmospheric turbulence
can induce random fluctuations of the intensity of the transmitted light beam \cite{Chakraborty}, which creates fading and complicates secure communication.
Note, however, that this fading is fundamentally different from multipath
fading and is more manageable from the standpoint of achieving secure
communication.

MIMO Poisson channels have received some interest lately (see \cite{Chakraborty2} and the references therein), and as has been done in the Gaussian setting, it would be interesting to see
the impact of having multiple antennas on the secrecy capacity in the Poisson regime.

We believe that the results derived in this paper and the tools used to derive them could be used to address these problems.
\newpage
\begin{appendices}
\section{An MMSE Proof for Lemma $6$}
In this first appendix, we provide an alternative proof for Lemma 6. This proof uses the link established in \cite{Guo} between the MMSE and the mutual information in Poisson channels. Note first that since $\mathbb{E}\int_{0}^{T}|X_{t}\ln X_{t}|<\infty$, we have that $I(X_{0}^{T};\mathcal{P}_{0}^{T}(A_{y} X_{0}^{T}+\lambda))$ is differentiable and Theorem 3 in \cite{Guo} states that
\begin{align}
&\frac{d}{d\lambda}I(X_{0}^{T};\mathcal{P}_{0}^{T}(A_{y} X_{0}^{T}+\lambda))=\nonumber\\&\int_{0}^{T}\mathbb{E}\{\ln(A_{y} X_{t}+\lambda)-\ln\langle A_{y} X_{t}+\lambda\rangle_{T}\}dt.
\end{align}
Notice now that
\begin{align}
&I(X_{0}^{T};\tilde{Y}_{0}^{T})-I(X_{0}^{T};Y_{0}^{T})=\nonumber\\&\int_{\lambda_{y}}^{\lambda_{y}+\tilde{\lambda}}\frac{d}{d\lambda}I(X_{0}^{T};\mathcal{P}_{0}^{T}(A_{y} X_{0}^{T}+\lambda))d\lambda.
\end{align}
Therefore
\begin{align}
  &I(X_{0}^{T};Y_{0}^{T})-I(X_{0}^{T};\tilde{Y}_{0}^{T})= \nonumber \\
  &\int_{\lambda_{y}}^{\lambda_{y}+\tilde{\lambda}}\!
\!\left(\!\!\int_{0}^{T}\!\!(\mathbb{E}\{\ln\langle A_{y}X_{t}+\lambda\rangle_{T}\}\!-\!\mathbb{E}\{\ln(A_{y} X_{t}+\lambda)\})dt\right)d\lambda.
\end{align}
Since the function $\ln(\cdot)$ is concave, using Jensen's inequality and the iterative conditioning property we have
\[\mathbb{E}\{\ln\langle A_{y}X_{t}+\lambda\rangle_{T}\}\leq\ln\mathbb{E}[\langle A_{y}X_{t}+\lambda\rangle_{T}]=\ln(A_{y}\mathbb{E}[X_{t}]+\lambda).\]
Making use of this inequality and the fact that $\lambda_{y}+\tilde{\lambda}=\frac{A_{y}}{A_{z}}\lambda_{z}$ we deduce that
\begin{align}
  &I(X_{0}^{T};Y_{0}^{T})-I(X_{0}^{T};\tilde{Y}_{0}^{T})\leq \nonumber \\
  &\!
\int_{0}^{T}\!\!\left(\int_{\lambda_{y}}^{\frac{A_{y}}{A_{z}}\lambda_{z}}\!\!\!\!\!\!\!\!\ln(A_{y}\mathbb{E}[X_{t}]+\lambda)d\lambda\!-\!\mathbb{E}\{\int_{\lambda_{y}}^{\frac{A_{y}}{A_{z}}\lambda_{z}}\!\!\!\!\!\!\!\!\ln(A_{y} X_{t}+\lambda)d\lambda\}\right)dt,
\end{align}
where we have also invoked Fubini's theorem to make the necessary exchanges between the integrals and the
expectation operator. The desired inequality is then obtained after some algebraic manipulations using the
elementary identity
\begin{equation}
\int\ln(A_{y}x+\lambda)d\lambda=(A_{y}x+\lambda)\ln(A_{y}x+\lambda)-\lambda.
\end{equation}

\section{An MMSE Proof for Lemma $7$}
Here we provide an alternative proof for Lemma 7. For ease of notations define $W_{t}=A_{y}X_{t}+\frac{A_{y}}{A_{z}}\lambda_{z}$. Using Theorem 4 in \cite{Guo} we obtain that
\begin{align}
\frac{d}{d\alpha}I(W_{0}^{T};\mathcal{P}_{0}^{T}(\alpha W_{0}^{T}))&=\int_{0}^{T}\mathbb{E}[W_{t}\ln(\alpha
W_{t})]dt-\int_{0}^{T}\mathbb{E}[\mathbb{E}[W_{t}|\mathcal{P}_{0}^{T}(\alpha W_{0}^{T})]\ln( \mathbb{E}[\alpha W_{t}|\mathcal{P}_{0}^{T}(\alpha W_{0}^{T})])]dt\nonumber\\
&=\int_{0}^{T}\mathbb{E}[W_{t}\ln W_{t}]dt-\int_{0}^{T}\mathbb{E}[\mathbb{E}[W_{t}|\mathcal{P}_{0}^{T}(\alpha W_{0}^{T})]\ln(
\mathbb{E}[W_{t}|\mathcal{P}_{0}^{T}(\alpha W_{0}^{T})])]dt,
\end{align}
where the second equality is obtained after some simplifications using the identity $\mathbb{E}[\mathbb{E}[W_{t}|\mathcal{P}_{0}^{T}(\alpha W_{0}^{T})]]=\mathbb{E}[W_{t}]$. Now by the convexity of the function $C(x)=x\ln(x)$, Jensen's inequality gives
\begin{align}
\mathbb{E}[C(\mathbb{E}[W_{t}|\mathcal{P}_{0}^{T}(\alpha W_{0}^{T})])]&\geq C(\mathbb{E}[\mathbb{E}[W_{t}|\mathcal{P}_{0}^{T}(\alpha W_{0}^{T})]])\nonumber\\
&=C(\mathbb{E}[W_{t}]).
\end{align}
It follows therefore that
\begin{align}
\frac{d}{d\alpha}I(W_{0}^{T};\mathcal{P}_{0}^{T}(\alpha W_{0}^{T}))\leq\int_{0}^{T}\mathbb{E}[W_{t}\ln W_{t}]dt-\int_{0}^{T}\mathbb{E}[W_{t}]\ln \mathbb{E}[W_{t}]dt.
\end{align}
Clearly we have that $I(W_{0}^{T};\mathcal{P}_{0}^{T}(\alpha W_{0}^{T}))=I(X_{0}^{T};\mathcal{P}_{0}^{T}(\alpha
W_{0}^{T}))$. Also we have that $\tilde{Y}_{0}^{T}=\mathcal{P}_{0}^{T}(W_{0}^{T})$ and $Z_{0}^{T}=\mathcal{P}_{0}^{T}(\frac{A_{z}}{A_{y}} W_{0}^{T})$. Consequently
\begin{align}
I(X_{0}^{T};\tilde{Y}_{0}^{T})-
I(X_{0}^{T};Z_{0}^{T})&=\int_{\frac{A_{z}}{A_{y}}}^{1}\frac{d}{d\alpha}I(W_{0}^{T};\mathcal{P}_{0}^{T}(\alpha
W_{0}^{T}))d\alpha.
\end{align}
Using the previous inequality, we conclude that
\begin{align}
I(X_{0}^{T};\tilde{Y}_{0}^{T})-
I(X_{0}^{T};Z_{0}^{T})&\leq\int_{\frac{A_{z}}{A_{y}}}^{1}\left(\int_{0}^{T}\mathbb{E}[W_{t}\ln W_{t}]dt-\int_{0}^{T}\mathbb{E}[W_{t}]\ln \mathbb{E}[W_{t}]dt\right)d\alpha\nonumber\\
&=(1-\frac{A_{z}}{A_{y}})\left(\int_{0}^{T}\mathbb{E}[(A_{y}X_{t}+\frac{A_{y}}{A_{z}}\lambda_{z})\ln(A_{y}X_{t}+\frac{A_{y}}{A_{z}}\lambda_{z})]dt\right.\nonumber\\&-\left.\int_{0}^{T}(A_{y}\mathbb{E}[X_{t}]+\frac{A_{y}}{A_{z}}\lambda_{z})\ln(A_{y}\mathbb{E}[X_{t}]+\frac{A_{y}}{A_{z}}\lambda_{z})dt\right).
\end{align}
After some simplifications, the last inequality gives the desired result, i.e.,
\begin{align}
I(X_{0}^{T};\tilde{Y}_{0}^{T})-
I(X_{0}^{T};Z_{0}^{T})\leq(\frac{A_{y}}{A_{z}}-1)\int_{0}^{T}\left(\mathbb{E}[\phi_{z}(X_{t})]-\phi_{z}(E[X_{t}])\right)dt.
\end{align}
\end{appendices}

\end{document}